\documentclass[pra,english,reprint,nofootinbib,aps,superscriptaddress,showpacs,showkeys]{revtex4-1}

\usepackage{babel,calc,amsmath,amsthm,amssymb,graphicx,subfigure,xcolor,
      longtable,braket,bm,cases,makecell}

\usepackage[T1]{fontenc}
\usepackage[unicode=true]{hyperref}
\hypersetup{
     colorlinks=true,       		
     linkcolor=blue,          	
     citecolor=red,            
     urlcolor=magenta,           	
}
\newtheorem*{theorem*}{Theorem}
\newtheorem{theorem}{Theorem}

\newtheorem*{corollary*}{Corollary}

\newtheorem*{lemma*}{Lemma}

\newtheorem*{proposition*}{Proposition}

\newtheorem*{conjecture*}{Conjecture}
\theoremstyle{definition}

\newtheorem*{definition*}{Definition}
\theoremstyle{remark}

\newtheorem*{remark*}{Remark}

\newcommand{\Eq}[1]{Eq.~\eqref{#1}}

\begin{document}
\title{$SO(4)$ Symmetry in Hydrogen Atom with Spin}

\author{Xing-Yan Fan}
\affiliation{Theoretical Physics Division, Chern Institute of Mathematics and LPMC, Nankai University, Tianjin 300071, People's Republic of China}

\author{Xiang-Ru Xie}
\affiliation{School of Physics, Nankai University, Tianjin 300071, People's Republic
      of China}

\author{Sheng-Ming Li}
\affiliation{School of Physics, Nankai University, Tianjin 300071, People's Republic
      of China}

\author{Jing-Ling Chen}
\email{chenjl@nankai.edu.cn}
\affiliation{Theoretical Physics Division, Chern Institute of Mathematics and LPMC, Nankai University, Tianjin 300071, People's Republic of China}

\date{\today}

\begin{abstract}
As the simplest atom in nature, the hydrogen atom has been explored thoroughly from the perspective of non-relativistic quantum mechanics to relativistic quantum mechanics. Among the research on hydrogen atom, its energy level is the most basic, which can be obtained more conveniently predicated on the $SO(4)$ symmetry than the wave-equation resolution. Moreover, ``spin'' is another indispensable topic in quantum mechanics, appearing as an intrinsic degree of freedom. In this work, we generalize the quantum Runge-Lenz vector to a spin-dependent one, and then extract a novel  Hamiltonian of hydrogen atom with spin based on the requirement of $SO(4)$ symmetry. Furthermore, the energy spectrum of hydrogen atom with spin potentials is also determined by the remarkable approach of $SO(4)$ symmetry. Our findings extend the ground of hydrogen atom, and may contribute to other complicated models based on hydrogen atom.
\end{abstract}


\maketitle
\section{Introduction}
After the birth of quantum mechanics (QM), the Schr\"odinger equation was used to address the problem of hydrogen atom (HA), which underlies the solution to other problem involving Coulomb interaction in QM, such as the problem of hydrogen-like atom \cite{GrinerQMHAtom}. Furthermore, the triumph of computing its energy spectrum
prompt physicists to research on other questions covering HA using quantum theory, for example, the low-energy scattering through HAs \cite{1962RMPBurke}, the perturbations of HA \cite{2010RMPEfstathiou}, or the long-range interaction of HAs at finite temperatures \cite{2024PRAZalialiutdinov}.

As is well known that symmetry is instrumental importantly in solving the HA problem \cite{2022A1CourSym}. Explicitly, the spherical symmetry helps us decompose the Schr\"odinger equation into the radial and angular parts, thus simplifying the complicated partial differential equation to its ordinary partner; certainly if one merely focuses on the energy spectrum of HA, a succinct method can be adopted concerning $SO(4)$ symmetry \cite{1988PauliSO4,GrinerSymSO4,1991Lange}. Thereafter, this hidden symmetry drew more attention \cite{1966RMPBANDER1,1966RMPBANDER2}. For instance, it was used to explicate the Stark states of HA as an irreducible representation of the $O(4)$ group \cite{1967PPSHughes}; more its subgroup reductions were considered in distinct bases \cite{1976SIAMJAM}; it was investigated even in virtue of the computer algebra system \cite{2021CPCSzriftgiser}.

Another essential conception ``spin'' is also connected with the symmetry in QM intimately. Under the circumstance of non-relativistic QM (NQM), the $SU(2)$ description for spin-1/2 system was revealed by Pauli \cite{1927Pauli}. Soon then Dirac finished the relativistic quantum-mechanical treatment for electron \cite{1928Dirac}; emphatically, the ``spinors'' as well as the Dirac matrices (satisfying Clifford algebra) appeared for the first time, which were then developed as a branch of geometry, dubbed spin geometry \cite{1989SpinGeo}. It is worth mentioning that the spin-based symmetry is not just a theoretical object, but an experimentally-observed property in condensed matter physics, e.g., the spin-group symmetry in magnetic materials \cite{2022PRXLiu}.

%

For a more prevalent taste of physicists, combining different symmetry will refresh their understanding of diverse models in QM. In this work, we advance the study of HA by pondering $SO(4)$ symmetry and ``spin'' simultaneously, in turn, we successfully attain a spin-dependent Hamiltonian for HA that still possesses the $SO(4)$ symmetry.

The article is organized as follows: In Sec. \ref{sec:review}, we make a brief review of $SO(4)$ symmetry by starting from the Hamiltonian of HA; In Sec. \ref{sec:inverse}, we consider an inverse problem, i.e., how to determine the Hamiltonian of the usual HA if one  starts from the quantum Runge-Lenz (QRL) vector; In Sec. \ref{sec:GenQRLVec}, we generalize the quantum Runge-Lenz vector with spin, and extract a Hamiltonian of HA with spin based on the requirement of $SO(4)$ symmetry; In Sec. \ref{sec:EHSpin}, the energy spectrum of HA with spin potentials is determined by the approach of $SO(4)$ symmetry; Conclusion and discussion are made in the Sec. \ref{sec:DisConcl}.

\section{Brief Review of $SO(4)$ Symmetry in Hydrogen Atom}\label{sec:review}

The Hamiltonian of HA is given by~\cite{GrinerQMHAtom}
            \begin{equation}\label{eq:ha-1}
                  H=\dfrac{\vec{p}^{\,2}}{2M} -\dfrac{\kappa}{r},
            \end{equation}
      where $\vec{p}=(p_x, p_y, p_z)$ is the linear momentum operator; $M=[m_{\rm e}m_{\rm p} /(m_{\rm e}+m_{\rm p})]$ is the reduced mass with $m_{\rm e}$ and $m_{\rm p}$ representing the mass of electron and proton, respectively; $\kappa=Z e^2$, and for HA, $Z=1$; $e$ depicts the absolute value of the charge involving an electron; $r=|\vec{r}|$ indicates the norm of coordinate operator $\vec{r}=(r_x, r_y, r_z)$ regarding the relative motion between electron and proton.

      One may define the orbital angular momentum operator as
             \begin{equation}
             \vec{\ell}=\vec{r}\times \vec{p},
            \end{equation}
      which satisfies the definition of angular momentum
             \begin{equation}
             \vec{\ell} \times \vec{\ell}={\rm i}\hbar \, \vec{\ell},
            \end{equation}
      or in a commutator form as
             \begin{equation}\label{eq:so3}
            \left[{\ell}_u, {\ell}_v\right]={\rm i}\hbar \, \epsilon_{uvw} {\ell}_w,
            \end{equation}
      with $\epsilon_{uvw}$ being the Levi-Civita symbol, and $\epsilon_{xyz}=\epsilon_{yzx}=\epsilon_{zxy}=+1$, $\epsilon_{zyx}=\epsilon_{yxz}=\epsilon_{xzy}=-1$.

      Due to
      \begin{eqnarray}
           && \left[{\ell}_u, r_v\right]={\rm i}\hbar \, \epsilon_{uvw} r_w,\;\;\;\;\left[{\ell}_u, p_v\right]={\rm i}\hbar \, \epsilon_{uvw} p_w,
            \end{eqnarray}
      it is easy to have
      \begin{equation}
             \left[\vec{\ell},\, r\right]=0, \;\;\; \left[\vec{\ell}, \,\vec{p}^{\,2}\right]=0,
      \end{equation}
      and then one proves that
       \begin{equation}
             \left[H, \,\vec{\ell}\, \right]=0.
      \end{equation}
Notably, Eq. (\ref{eq:so3}) implies that a HA system possesses $SO(3)$ symmetry.

      However, the HA can have a hidden symmetry of $SO(4)$ owing to the existence of the quantum Runge-Lenz (QRL) vector, defined by~\cite{1988PauliSO4,GrinerSymSO4}
            \begin{eqnarray}\label{eq:DefR}
                  && \vec{R}=\dfrac{1}{2M} \left(\vec{p}\times\vec{\ell}
                              -\vec{\ell}\times\vec{p}\right)
                        -\kappa\,\dfrac{\vec{r}}{r},
            \end{eqnarray}
      which satisfies
      \begin{eqnarray}
       && \vec{R}\times \vec{R} = \dfrac{-2 H}{M}{\rm i}\hbar \, \vec{\ell}, \label{eq:RREl}
      \end{eqnarray}
      and
      \begin{subequations}
                  \begin{eqnarray}
                        && \left[H, \vec{R}\right]=0, \label{eq:RHComut} \\
                        && \vec{\ell}\cdot\vec{R}=\vec{R}\cdot\vec{\ell}
                              =0, \label{eq:RLLR0} \\
                        && \vec{R}^{\, 2} =\dfrac{2 H}{M} \left(\vec{\ell}^{\,2} +\hbar^2\right)
                              +\kappa^2, \label{eq:R2H} \\
                        && \left[{\ell}_u,\, {R}_v\right]={\rm i}\hbar \, \epsilon_{uvw}
                                    R_w. \label{eq:ElR}
                  \end{eqnarray}
       \end{subequations}
       Since $\{\vec{\ell}, \vec{R}\}$ satisfy the following $SO(4)$ algebraic relations (up to a re-scaling factor $\sqrt{M/(-2H)}$ multiplied together $\vec{R}$)
      \begin{subequations}
      \begin{eqnarray}
         &&   \left[{\ell}_u, {\ell}_v\right]={\rm i}\hbar \, \epsilon_{uvw} {\ell}_w, \label{eq:SO4H-1b} \\
         &&   \left[{\ell}_u, {R}_v\right]={\rm i}\hbar \, \epsilon_{uvw} {R}_w, \label{eq:SO4H-2b}\\
         &&   \left[{R}_u, {R}_v\right]=\frac{-2 H}{M}{\rm i}\hbar \, \epsilon_{uvw} {\ell}_w, \label{eq:SO4H-3b}
      \end{eqnarray}
      \end{subequations}
      and thus the HA system possesses $SO(4)$ symmetry.

      \section{An Inverse Problem}\label{sec:inverse}

      A HA is a real physical system existing in nature. The Hamiltonian (\ref{eq:ha-1}) of a HA can be recast to
             \begin{equation}\label{eq:ha-2}
                  H=\dfrac{\vec{p}^{\,2}}{2M} +V_{\rm cp}(r),
            \end{equation}
      where
       \begin{equation}\label{eq:ha-3}
                  V_{\rm cp}(r) =-\dfrac{\kappa}{r}
            \end{equation}
      represents the well-known Coulomb potential (CP). In the previous section, we have considered a \emph{forward problem}: Given the Hamiltonian of a HA, can one determine its underlying symmetry? The answer is positive. The underlying symmetry is just $SO(4)$ symmetry, as shown in Eqs. (\ref{eq:SO4H-1b})-(\ref{eq:SO4H-3b}).

      In this section, let us consider an \emph{inverse problem}: Given the underlying $SO(4)$ symmetry, can one construct the Hamiltonian of a HA (or derive the CP)? The answer is still positive. The key point is the QRL vector; if one focuses on Eq. (\ref{eq:RREl}), he will find that the Hamiltonian $H$ can be extracted from $\vec{R}\times \vec{R}$ by requiring the result proportional to $H\,\vec{\ell}$. The following is the procedure of construction.

      Without lose of generality, let us define the QRL vector in the following form:
      \begin{eqnarray}\label{eq:DefR-1a}
                  && \vec{R}=\dfrac{1}{2M} \left(\vec{p}\times\vec{\ell}
                              -\vec{\ell}\times\vec{p}\right)
                        + f(r)\, \vec{r},
            \end{eqnarray}
      where $f(r)$ is a certain function depending on $r$, and we only require
       \begin{eqnarray}\label{eq:DefR-1b}
                  && [\vec{\ell}, f(r)] = 0.
            \end{eqnarray}
      If one compares Eq. (\ref{eq:DefR}) and Eq. (\ref{eq:DefR-1a}), one may notice that
      \begin{eqnarray}\label{eq:DefR-1c}
                  && f(r) =-\dfrac{\kappa}{r}.
            \end{eqnarray}
      However, suppose we do not know it in the beginning, thus temporarily treating $f(r)$ as an unknown function.

      By direct calculation, we easily have
      \begin{eqnarray}
                  && [{\ell}_u,\ {\ell}_v]={\rm i}\hbar\, \epsilon_{uvw}\, {\ell}_w, \label{eq:SO4H-1c}\\
                 && [{\ell}_u,\ {R}_v]={\rm i}\hbar\, \epsilon_{uvw}\, {R}_w.\label{eq:SO4H-2c}
      \end{eqnarray}
      The underlying $SO(4)$ symmetry requires the following commutative relation
       \begin{eqnarray}
               && [{R}_u,\ {R}_v] \propto {\rm i}\hbar\, \epsilon_{uvw}\, {\ell}_w,
      \end{eqnarray}
      i.e., $[{R}_u,\ {R}_v]$ is required to be proportional to $\ell_w$, and the proportional coefficient is related to the Hamiltonian $H$ of a HA. To reach this purpose, let us compute the commutator $[{R}_x,\ {R}_y]$, or $\vec{R}\times \vec{R}$.

      The result gives
      \begin{eqnarray}\label{eq:RCrosR-0a}
                   \vec{R}\times\vec{R}
                 & =&-{\rm i}\hbar\dfrac{2}{M} \biggl(\dfrac{\vec{p}^{\, 2}}{2\,M}
                        +\dfrac{1}{2} \Bigl\{\big[\vec{\nabla} f(r)\big]\cdot\vec{r}+3\,f(r)
                              \Bigr\}\biggr)\vec{\ell} \notag \\
                  &=& \dfrac{-2 H}{M}{\rm i}\hbar \, \vec{\ell},
            \end{eqnarray}
     based on which we can immediately extract a Hamiltonian as
            \begin{eqnarray}
                  & H=\dfrac{\vec{p}^{\,2}}{2M} +{\cal V}(r),
            \end{eqnarray}
            with the potential function
             \begin{eqnarray}
                   {\cal V}(r)&=&\dfrac{1}{2} \Bigl\{
                        \big[\vec{\nabla} f(r)\big]\cdot\vec{r}+3\,f(r)\Bigr\} \nonumber\\
                   &=&\dfrac{1}{2} \Bigl\{
                       r f'(r) +3\,f(r)\Bigr\}.
            \end{eqnarray}

     It is easy to prove that
            \begin{eqnarray}
                  \left[\vec{\ell}, H\right]=0.
            \end{eqnarray}
     Furthermore, we use the following condition
            \begin{eqnarray}
                  \left[\vec{R}, H\right]=0
            \end{eqnarray}
     to determine the function $f(r)$, for $\vec{R}$ is the conservation quantity of the Hamiltonian system. One can have
             \begin{eqnarray}
       [\vec{R}, H]&=& \dfrac{-{\rm i}\hbar}{M} \left[\vec{\nabla}\left[{\cal V}(r)-f(r)\right]\cdot \vec{p}\right] \vec{r}\nonumber\\
                &&+\dfrac{{\rm i}\hbar}{M} \left[\left(\vec{\nabla}{\cal V}(r)\right)\cdot \vec{r}+f(r)\right] \vec{p}\nonumber\\
                &&+\dfrac{\hbar^2}{2 M}  \left(\vec{\nabla}^2 f(r)\right) \vec{r},
      \end{eqnarray}
      and thus
      \begin{subequations}
      \begin{eqnarray}
         && \vec{\nabla}\left[{\cal V}(r)-f(r)\right]=0, \\
         &&  \left(\vec{\nabla}{\cal V}(r)\right)\cdot \vec{r}+f(r)=0,  \\
         &&  \vec{\nabla}^2 f(r)=0,
      \end{eqnarray}
      \end{subequations}
      which lead to
      \begin{eqnarray}\label{eq:Fr1OvrR}
                  &&  {\cal V}(r) =f(r) =\dfrac{C}{r}.
     \end{eqnarray}
     Consequently, by choosing $C=-\kappa$ we have successfully constructed the Hamiltonian of a HA by starting from the QRL vector.

     This significant finding of the inverse problem is very helpful for studying a more complicated Hamiltonian under the constraint of $SO(4)$ symmetry. In the next section,  based on $SO(4)$ symmetry we will construct the Hamiltonian of HA with spin.


\section{Constructing the Hamiltonian of Hydrogen Atom with Spin}\label{sec:GenQRLVec}

To construct the Hamiltonian of HA with spin, we need to generalize the orbital angular momentum $\vec{\ell}$ and the spin-independent QPL vector $\vec{R}$ to the spin-dependent ones, namely, $\vec{\ell}\Rightarrow \vec{J}$ and $\vec{R} \Rightarrow \vec{\mathcal{R}}$, where $\vec{J}$ and $\vec{\mathcal{R}}$ depend on the spin angular momentum operator $\vec{S}$, which satisfies the definition of angular momentum
             \begin{equation}
             \vec{S} \times \vec{S}={\rm i}\hbar \, \vec{S},
            \end{equation}
      or in a commutator form as
             \begin{equation}
            \left[{S}_u, {S}_v\right]={\rm i}\hbar \, \epsilon_{uvw} {S}_w.
            \end{equation}
Our main result is the following theorem:

\begin{theorem}Let
         \begin{subequations}
         \begin{eqnarray}
         &&    \vec{J}=\vec{\ell}+ \mu \vec{S}, \label{eq:GJ}\\
         &&  \vec{\cal R} =\dfrac{1}{2\,M} \big(\vec{\Pi}\times\vec{J}
            -\vec{J}\times\vec{\Pi}\big)+h(r)\,\vec{r},\label{eq:GQRL}
         \end{eqnarray}
      \end{subequations}
and then they satisfy the $SO(4)$  commutation relations (up to a re-scaling factor multiplied together $\vec{\cal R}$)
\begin{subequations}
      \begin{eqnarray}
            && [J_u,\ J_v]={\rm i}\hbar\, \epsilon_{uvw} J_w,\label{eq:so4-1a} \\
            && [J_u,\ {\cal R}_v]={\rm i}\hbar\, \epsilon_{uvw}
            {\cal R}_w, \label{eq:so4-1b} \\
            &&  [{\cal R}_u,\ {\cal R}_v]\propto{\rm i}\hbar \,
            \epsilon_{uvw} {\cal H}\,J_w, \label{eq:so4-1c}
      \end{eqnarray}
\end{subequations}
and
\begin{subequations}
      \begin{eqnarray}
            && [\vec{J},\, {\cal H}]=0, \label{eq:JHComt} \\
            && [\vec{\cal R},\, {\cal H}]=0. \label{eq:RHComt}
      \end{eqnarray}
\end{subequations}
Here the Hamiltonian of HA with spin reads
\begin{eqnarray}\label{eq:H}
      && {\cal H}=\dfrac{\vec{\Pi}^2}{2\,M} +{\cal V}(r),
\end{eqnarray}
with
\begin{subequations}
\begin{eqnarray}
      &&      \vec{\Pi}=\vec{p}-\vec{A}, \label{eq:A-1}\\
      &&  \vec{A}=\mu\dfrac{(\vec{r}\times\vec{S})}{r^2}, \label{eq:A-2}\\
      &&       h(r)=\dfrac{k_1}{r} +k_2\, \mu \,\dfrac{\vec{r}\cdot\vec{S}}{r^2},
            \label{eq:A-3}\\
      && {\cal V}(r)=h(r)+\mu\dfrac{1}{2\,M} \dfrac{(\vec{r}\cdot\vec{S})^2}{r^4},
            \label{eq:A-4}
\end{eqnarray}
\end{subequations}
$k_1$ and $k_2$ are two constants, and $\mu=0, 1$.
\end{theorem}

\begin{proof} We divide five steps to complete the proof.

\emph{Step 1: Construction of $\vec{J}$ and Proving Eq. (\ref{eq:so4-1a}).}---Since the orbital angular momentum $\vec{\ell}$ and the spin operator $\vec{S}$ are defined in the external and internal spaces respectively, they are mutually commutative. If we define the ``total'' angular momentum operator as
\begin{eqnarray}\label{eq:GJ-1}
&&    \vec{J}=\vec{\ell}+ \vec{S},
\end{eqnarray}
then one immediately proves Eq. (\ref{eq:so4-1a}). Here we have designed $\vec{J}$ in the form of Eq. (\ref{eq:GJ}) with the parameter $\mu$, to treat the HAs with and without spin in a unified way; namely, $\vec{J}$ reduces to $\vec{\ell}$ if $\mu=0$, and becomes Eq. (\ref{eq:GJ-1}) if $\mu=1$.

\emph{Step 2: Construction of $\vec{\cal R}$ and Proving Eq. (\ref{eq:so4-1b}).}---Similar to Ref. \cite{2023FRSpinV}, one can recast $\vec{J}$ to the following form
\begin{eqnarray}\label{eq:GJ-2}
    \vec{J} &= &\vec{\ell}+ \mu \vec{S}\nonumber\\
    &=& \vec{r}\times (\vec{p}-\vec{A}) + \mu\, \vec{e}_r S_r \nonumber\\
    &=& \vec{r}\times \vec{\Pi} + \mu\, \vec{e}_r S_r;
\end{eqnarray}
then one extracts ``canonical'' momentum $\vec{\Pi}$ of the electron in HA as in Eq. (\ref{eq:A-1}), and the spin vector potential $\vec{A}$ induced by spin operator $\vec{S}$ of the proton in HA as Eq. (\ref{eq:A-2}). Here the operator
\begin{eqnarray}\label{eq:sr-1}
    \hat{S}_r= \hat{r}\cdot \vec{S},
\end{eqnarray}
and $\hat{r}=\vec{e}_r=\vec{r}/r$.

After performing the following replacement
\begin{eqnarray}
    \vec{p} \Rightarrow \vec{\Pi}, \quad \vec{\ell} \Rightarrow \vec{J},\quad f(r)\Rightarrow h(r),
\end{eqnarray}
the QPL vector $\vec{R}$ is turned into the generalized QPL vector $\vec{\cal R}$ as shown in Eq. (\ref{eq:GQRL}), with $h(r)$ being a function of $r$ and $(\vec{r}\cdot\vec{S})$ satisfying the commutation relation
\begin{eqnarray}
   [\vec{J}, \, h(r)]=0.
\end{eqnarray}
Due to the relations
\begin{subequations}
      \begin{eqnarray}
            && [J_u,\ J_v]={\rm i}\hbar\, \epsilon_{uvw} J_w, \label{eq:JuJv} \\
            && [J_u,\ \Pi_v]={\rm i}\hbar\, \epsilon_{uvw} \Pi_w, \\
            && [J_u,\ (\vec{\Pi}\times\vec{J})_v]={\rm i}\hbar\, \epsilon_{uvw} (\vec{\Pi}\times\vec{J})_w, \\
            && [J_u,\ (\vec{J}\times \vec{\Pi})_v]={\rm i}\hbar\, \epsilon_{uvw} (\vec{J}\times\vec{\Pi})_w, \\
            && [J_u,\ r_v]={\rm i}\hbar\, \epsilon_{uvw} r_w, \label{eq:JuRv}
      \end{eqnarray}
\end{subequations}
one directly proves Eq. (\ref{eq:so4-1b}).

\emph{Step 3: Proving Eq. (\ref{eq:so4-1c}).}---After careful deductions, one has
\begin{widetext}
       \begin{eqnarray}\label{eq:RCrosR-1a}
                   \vec{\cal R}\times\vec{\cal R}
                  &=& {\rm i}\hbar\,\dfrac{-2}{M} \Bigg\{\dfrac{\vec{\Pi}^2}{2\,M}
                        +\dfrac{1}{2} \bigg[-\dfrac{1}{{\rm i}\hbar} \big[
                              \vec{\Pi},\ h(r)\big]\cdot\vec{r}
                        +3\,h(r)+\dfrac{\mu}{M}
                              \dfrac{(\vec{r}\cdot\vec{S})^2}{r^4}\bigg]\Bigg\}\vec{J}
                        -\dfrac{\mu}{M} \biggl\{\big[\vec{\Pi},\ h(r)\big]
                              -{\rm i}\hbar\dfrac{\vec{r}}{r^2} h(r)
                              \biggr\}(\vec{r}\cdot\vec{S}).
            \end{eqnarray}

          To fulfill the requirement of Eq. (\ref{eq:so4-1c}), one must have
      \begin{subequations}
            \begin{eqnarray}
             && {\cal H}=\dfrac{\vec{\Pi}^2}{2\,M} +{\cal V}(r),\\
                  && {\cal V}(r)=\dfrac{1}{2} \bigg[-\dfrac{1}{{\rm i}\hbar} \big[
                              \vec{\Pi},\ h(r)\big]\cdot\vec{r}
                        +3\,h(r)+\dfrac{\mu}{M}
                              \dfrac{(\vec{r}\cdot\vec{S})^2}{r^4}\bigg], \label{eq:Vr} \\
                  && \dfrac{\mu}{M} \biggl\{\big[\vec{\Pi},\ h(r)\big]
                             -{\rm i}\hbar\dfrac{\vec{r}}{r^2} h(r)
                              \biggr\}(\vec{r}\cdot\vec{S})=0. \label{eq:Constr}
            \end{eqnarray}
      \end{subequations}
  \end{widetext}
  For $\mu=0$, Eq. (\ref{eq:RCrosR-1a}) reduces to Eq. (\ref{eq:RCrosR-0a}), and thus Eq. (\ref{eq:so4-1c}) is proved.
   For $\mu=1$, from Eq. (\ref{eq:Constr}) we obtain
         \begin{eqnarray}
             \big[\vec{\Pi},\ h(r)\big] -{\rm i}\hbar\dfrac{\vec{r}}{r^2} h(r)
                              =0.  \label{eq:Constr-1}
            \end{eqnarray}
 By setting
 \begin{eqnarray}
          &&   h(r)= \xi_1(r) + \xi_2(r)\left(\vec{r}\cdot\vec{S}\right),
            \end{eqnarray}
 one obtains the solution of $h(r)$ as shown in Eq. (\ref{eq:A-3}) with $\xi_1(r)=k_1/r$ and $\xi_2(r)=\mu k_2/r^2 $, and hence Eq. (\ref{eq:so4-1c}) is proved.
After substituting Eq. (\ref{eq:Constr}) into Eq. (\ref{eq:Vr}), one has the potential function ${\cal V}$ as in Eq. (\ref{eq:A-4}).

\emph{Step 4: Proving Eq. (\ref{eq:JHComt}).}---
Due to the relations
\begin{subequations}
      \begin{eqnarray}
            && [\vec{J},\, \vec{\Pi}^2]=0, \label{eq:JPi2} \\
            && [\vec{J},\, r]=0, \\
             && [\vec{J},\, \vec{r}\cdot\vec{S}]=0, \label{eq:JRS}
            \end{eqnarray}
\end{subequations}
one directly proves Eq. (\ref{eq:JHComt}).

\emph{Step 5: Proving Eq. (\ref{eq:RHComt}).}---
Due to \Eq{eq:JuJv}-\Eq{eq:JuRv}, \Eq{eq:JPi2}-\Eq{eq:JRS}, and the following relations
\begin{subequations}
      \begin{eqnarray}
        && [\Pi_u,\ r_v]=-{\rm i}\hbar\,\delta_{uv},\\
          && \bigg[\vec{\Pi},\ \dfrac{1}{r}\bigg]={\rm i}\hbar\dfrac{\vec{r}}{r^3}, \\
            && \bigg[\vec{\Pi},\ \dfrac{1}{r^2}\bigg]={\rm i}\hbar\dfrac{2\,\vec{r}}{r^4},
                  \\
            && [\Pi_u,\, \Pi_v]={\rm i}\hbar\,\epsilon_{uvw} {\cal B}_w, \\
            && \vec{\cal B}={\rm i}\hbar\,\mu(\mu-2
                  )\dfrac{\vec{r}\cdot\vec{S}}{r^4} \vec{r}, \\
            && [\vec{\Pi},\, \vec{\Pi}^2]={\rm i}\hbar\big[(\vec{\Pi}\times\vec{\cal B})
                  -(\vec{\cal B}\times\vec{\Pi})\big], \\
            && \vec{\Pi}\times\vec{J}+\vec{J}\times\vec{\Pi}=2 {\rm i}\hbar\, \vec{\Pi}, \\
             && \bigg[\vec{\Pi},\ \dfrac{\vec{r}\cdot\vec{S}}{r^2}\bigg]
                  ={\rm i}\hbar(\mu-1)\dfrac{\vec{S}}{r^2}
                        +{\rm i}\hbar(2-\mu)\frac{\vec{r}\cdot\vec{S}}{r^4} \vec{r},
      \end{eqnarray}
\end{subequations}
with $\delta_{uv}$ denoting the Kronecker delta symbol, one can directly prove \Eq{eq:RHComt}.
\end{proof}

\section{Energy Spectrum of Hydrogen Atom with Spin}\label{sec:EHSpin}

In this section, we come to solve the energy spectrum of HA with spin based on the approach of $SO(4)$ symmetry. We restrict our problem to the Hilbert space of Hamiltonian ${\cal H}$'s eigenstates, where we can replace ${\cal H}$ by its eigenvalue $E$, and then \Eq{eq:so4-1c} becomes
\begin{eqnarray}
   &&   [{\cal R}_u,{\cal R}_v]=\dfrac{-2E}{M} {\rm i}\hbar\,\epsilon_{uvw} \,J_w.
\end{eqnarray}
If we define a re-scaling operator as
\begin{eqnarray}
      \vec{\cal R}^{\prime}=\sqrt{\dfrac{M}{-2E}}\, \vec{\cal R},
\end{eqnarray}
then $\vec{J}$ and $\vec{\cal R}'$ fit the standard SO(4) algebra.
Let us further define two new quantities as
 \begin{subequations}
\begin{eqnarray}
     && \vec{W}=\dfrac{1}{2}(\vec{J}+\vec{\cal R}'),\\
    &&  \vec{K}=\dfrac{1}{2}(\vec{J}-\vec{\cal R}').
\end{eqnarray}
 \end{subequations}
Then they satisfy the following commutators
 \begin{subequations}
\begin{eqnarray}
      &&[W_u,\ W_v] = {\rm i}\hbar\,\epsilon_{uvw} W_w, \\
      &&[K_u,\ K_v] = {\rm i}\hbar\,\epsilon_{uvw} K_w, \\
      &&[W_u,\ K_v] = 0,
\end{eqnarray}
 \end{subequations}
which implies that $\vec{W}$ and $\vec{K}$ construct two independent $SU(2)$ algebras \cite{GrinerSymSO4}; hence the eigenvalues are given as $\vec{W}^2=w(w+1)\hbar^2$, and $\vec{K}^2=k(k+1)\hbar^2$, with $k,w=0,1/2,1,3/2,\dots$.

We can have
 \begin{subequations}
\begin{eqnarray}
      \vec{W}^2 &=& \dfrac{1}{4}(\vec{J}^2+\vec{{\cal R}'}^2+2\vec{J}\cdot\vec{{\cal R}'}),
            \\
      \vec{K}^2 &=& \dfrac{1}{4}(\vec{J}^2+\vec{{\cal R}'}^2-2\vec{J}\cdot\vec{{\cal R}'}).
\end{eqnarray}
 \end{subequations}
Afterwards, by using
\begin{subequations}
\begin{eqnarray}
      && \vec{J}\cdot\vec{{\cal R}'}=\vec{{\cal R}'}\cdot\vec{J}
            =\mu\sqrt{\dfrac{M}{-2E}}\,h(r) (\vec{r}\cdot\vec{S}), \\
      && \vec{{\cal R}'}^2 =\left[\mu\dfrac{(\vec{r}\cdot\vec{S})^2}{r^2}
            -\vec{J}^{\,2} -\hbar^2\right] -\dfrac{M}{2E} \big[h(r)\big]^2 r^2,
\end{eqnarray}
\end{subequations}
we get
\begin{subequations}
      \begin{eqnarray}
            \vec{W}^2+\vec{K}^2  &=& \dfrac{1}{2} \bigg[
                       {\mu}\dfrac{(\vec{r}\cdot\vec{S})^2}{r^2} -\hbar^2
                        -\dfrac{M}{2E} \big[h(r)\big]^2 r^2\bigg],
                  \label{eq:I2+K2}\\
            \vec{W}^2-\vec{K}^2 &=& \mu\sqrt{\dfrac{M}{-2E}}\,h(r)\,(\vec{r}\cdot\vec{S}),\label{eq:I2-K2}
      \end{eqnarray}
\end{subequations}

Furthermore, one can verify that
\begin{eqnarray}
      \left[\hat{S}_r, \ {\cal H}\right]=0,
\end{eqnarray}
and therefore $\hat{S}_r$ is also a conserved quantity for the Hamiltonian system, and we denote its eigenvalue as $s_r\hbar$. In the case of
\begin{eqnarray}
      \vec{S}=\frac{\hbar}{2}\vec{\sigma}
\end{eqnarray}
i.e., the spin-$1/2$ operator, one has $S_r^2 =(\hbar^2 /4)\openone$, where $\vec{\sigma}$ and $\openone$ express the vector of Pauli matrices and identity matrix respectively, which means that $s_r=\pm 1/2$.

After replacing the operators by their corresponding eigenvalues, \Eq{eq:I2+K2} and \Eq{eq:I2-K2} become
\begin{widetext}
      \begin{subequations}
            \begin{eqnarray}
                  && w(w+1)\hbar^2 +k(k+1)\hbar^2 =\dfrac{1}{2} \bigg[
                        {\mu}(s_r \hbar)^2 -\hbar^2
                        -\dfrac{M}{2E} (k_1 +{\mu}\,k_2 s_r \hbar)^2
                        \bigg], \label{eq:i2+k2} \\
                  && w(w+1)\hbar^2 -k(k+1)\hbar^2 =\mu\sqrt{\dfrac{M}{-2E}}\,\big(
                        k_1 +{\mu}\,k_2 s_r \hbar\big)s_r \hbar.
                        \label{eq:i2-k2}
            \end{eqnarray}
      \end{subequations}
\end{widetext}
There is no solution for the energy $E$ if $k_1 +{\mu}\,k_2 s_r \hbar=0$. For
$k_1+k_2 s_r \hbar \neq 0$, one has the energy spectrum of HA with spin as
\begin{eqnarray}\label{eq:K1K2}
      E\equiv E_n &=&-\dfrac{M}{2\hbar^2} \dfrac{(k_1 +\mu\,k_2s_r \hbar)^2}{\big(n \pm \mu\,s_r \big)^2},
\end{eqnarray}
with $n=2\,w+1=1,2, 3,....$ For $\mu=0$,
\begin{eqnarray}
     E_n &=&-\dfrac{M k_1^2}{2\hbar^2} \frac{1}{n^2},
\end{eqnarray}
which reduces to the energy spectrum of the usual HA.

\section{Conclusion and Discussion}\label{sec:DisConcl}


To sum up, we have written down the spin-dependent QRL vector, constructed the Hamiltonian of HA with spin under the requirement of $SO(4)$ symmetry. The energy spectrum of the novel hamiltonian system is also determined by the approach of $SO(4)$ symmetry, which reduces to the usual HA without considering spin. Our computations extend the investigation of HA from the perspective of spin potential, which may be applied to HA-based models in atomic physics and condensed matter physics.

Considering the relativistic treatments for HA with $SO(4)$ approach have been fully studied, such as \cite{2008PRAChen,2008PRAZhang,2023APEremko}, it is our next object to explore the energy spectrum of HA with spin potential under the framework of relativistic quantum mechanics. Certainly, the connection of HA and harmonic oscillator is intimate (see e.g., \cite{2020JPABars}), so our future work also includes the quantum harmonic oscillator with spin potential. Finally, we hope our consideration of HA with spin potential can be realized experimentally, or can be achieved through quantum simulation \cite{2014QSim}, in that a more complex model (i.e., Rydberg manifold \cite{2022RydbQSim}) under $SO(4)$ symmetry has been simulated successfully.


\vspace{5mm}

\begin{acknowledgments}
   J.L.C. was supported by the National Natural Science Foundations of China (Grant Nos. 12275136 and 12075001) and the 111 Project of B23045. X.Y.F. was supported by the Nankai Zhide Foundations.
\end{acknowledgments}
X.Y.F. and X.R.X. contributed equally to this work.
\onecolumngrid
\twocolumngrid

\end{document}